\documentclass[11pt]{amsart}
\usepackage{mathrsfs} 
\usepackage[margin=1.25in]{geometry}
\usepackage{ifthen}
\usepackage{graphicx}
\usepackage{epsfig}
\usepackage{dsfont, amssymb}
\usepackage{amsfonts,dsfont,amssymb,amsthm,stmaryrd,bbm}
\usepackage{amsmath}
\usepackage{xcolor}

\usepackage[toc,page]{appendix} 
\usepackage{hyperref,mathtools}
\hypersetup{colorlinks=true,pdftitle=""%,pdftex
}

\let\originallesssim\lesssim
\let\originalgtrsim\gtrsim

\DeclareRobustCommand{\lesssim}{%
  \mathrel{\mathpalette\lowersim\originallesssim}%
}
\DeclareRobustCommand{\gtrsim}{%
  \mathrel{\mathpalette\lowersim\originalgtrsim}%
}

\makeatletter
\newcommand{\lowersim}[2]{%
  \sbox\z@{$#1<$}%
  \raisebox{-\dimexpr\height-\ht\z@}{$\m@th#1#2$}%
}
\makeatother

\newtheorem{thm}{Theorem}[section]

\newtheorem{remark}[thm]{Remark}

\newtheorem{prop}[thm]{Proposition}
\newtheorem{coro}[thm]{Corollary}

\newtheorem{defn}[thm]{Definition}

\newcommand\independent{\protect\mathpalette{\protect\independent}{\perp}} 
\def\independent#1#2{\mathrel{\rlap{$#1#2$}\mkern2mu{#1#2}}}

\newcommand{\R}{\mathbb{R}}

\def\E{{\mathbb{E}}}

\def\phi{\varphi}

\def\bee{\begin{eqnarray*}}
\def\ene{\end{eqnarray*}}
\usepackage{mathtools}

\begin{document}
\title{A Quantitative Entropy Power Inequality for Dependent Random Vectors}
\author{Mokshay Madiman, James Melbourne, and Cyril Roberto }
\date{\today}

\maketitle

\begin{abstract}
  The entropy power inequality  for independent random vectors is a foundational result of information theory, with deep connections to probability and geometric functional analysis. Several extensions of the entropy power inequality have been developed for settings with dependence, including by Takano, Johnson, and Rioul. We extend these works by developing a quantitative version of the entropy power inequality for dependent random vectors. A notable consequence is that an entropy power inequality stated using conditional entropies holds for random vectors whose joint density is log-supermodular.
\end{abstract}
  



% \fi
%\IEEEpeerreviewmaketitle

\section{Introduction}
\label{sec:intro}

When a random vector $X$ taking values in $\R^d$ has density $f$, the {\it entropy} of $X$ is 
\begin{equation}\label{eq:SNE}
h(X)=h(f):= -\int f(x)\log f(x)dx= \E[-\log f(X)] .
\end{equation}
This quantity is sometimes called the Boltzmann-Shannon or Boltzmann-Gibbs entropy, or the differential entropy.
In dimension $d$, the {\it entropy power} of $X$ is 
$
N(X)=e^{\frac{2h(X)}{d}}.
$
As is usual, we abuse notation and write $h(X)$ and $N(X)$,
even though these are functionals depending only on the density
of $X$ and not on its random realization.
The entropy power $N(X)\in [0,\infty]$ can be thought of as a ``measure of randomness''. 
It is an (inexact) analogue of volume: 
if $U_A$ is uniformly distributed on a bounded Borel set $A$, then it is easily checked that 
$h(U_A)=\log |A|$ and hence $N(U_A)=|A|^{2/d}$.
The reason for this particular definition of entropy power is that the ``correct'' comparison is not to uniforms but to Gaussians: observe that when $Z\sim N(0,\sigma^{2}I)$
(i.e., $Z$ has the Gaussian distribution with mean 0 and covariance matrix that is a multiple of the identity), 
the entropy power of $Z$ is $N(Z)=(2\pi e) \sigma^{2}$.
Thus the entropy power of $X$ is-- up to a universal constant--
the variance of the isotropic normal that has the same entropy as $X$,
i.e., if $Z\sim N(0,\sigma_Z^{2}I)$ and $h(Z)=h(X)$, then
$$
N(X)=N(Z)=(2\pi e) \sigma_{Z}^{2} .
$$
Looked at this way, $N(X)/(2\pi e)$ is the ``effective variance'' of the random vector $X$.
%exactly as volume raised to $1/d$ is the effective radius of a set.

The entropy power inequality (EPI) states that for any two independent random vectors $X$ and $Y$ in $\R^d$ such that the entropies of $X, Y$ and $X+Y$ exist, 
$$
N(X+Y) \geq N(X) + N(Y) .
$$
The EPI was stated by Shannon \cite{Sha48} with an incomplete proof; the first complete proof
was provided by Stam in his Ph.D. thesis \cite{Sta59:phd} (see also \cite{Sta59}). The EPI plays an important role in Information Theory, where it first arose and was used (first by Shannon, and later by many others) to prove statements about the fundamental limits of communication over various models of communication channels. It has also been recognized as a very useful inequality in Probability Theory, with close connections to the logarithmic Sobolev inequality for the Gaussian distribution as well as to the Central Limit Theorem. The history of the EPI and its connections to many other inequalities are described in the surveys \cite{DCT91,FMMZ18,Gar02,MMX17:0} and the connections to the Central Limit Theorem in \cite{Joh04:book, MR25}.

The EPI has been generalized and extended in several directions.  For instance, there has been considerable recent progress in the understanding of R\'enyi entropy analogs of the EPI (see \cite{BC14, BC15:1,BM17, FGR24, Gavalakis24, LMM20, MM19, melbourne2025minimum, melbourne2022quantitative, RS16,  Rio18, Tao10, WM14}), using which fruitful connections have been made with other branches of mathematics \cite{BNT16, Li18:1,  MK18, MMX17:1, MP23, MT20, melbourne2025renyi}.  Discrete analogs have been explored in \cite{ALM17, BMM22, HAT14, HV03, JY10, MMR23,  MWW19, MWW21, marsiglietti2025note, WM15:isit}.   Matrix extensions and connections of the EPI with powerful functional inequalities like the Brascamp-Lieb inequality have been discovered \cite{ AJN22, CL10, LCCV18, ZF93}.  Further strengthenings, sharpenings, and generalizations can be found in \cite{ABBN04:1, Cou18, CFP18, MB07,  MG19, MNT20, Tos15:1}, while so-called reverse entropy power inequalities for special classes of distributions have been obtained in \cite{BM11:cras, BM12:jfa, BM13}.

The aim of this paper is to deliver an EPI for dependent summands.  Several previous attempts have been made towards this same goal, see \cite{CS91,  Joh04:1, Joh05, Tak96, Tak98}, which we now briefly recollect. This paper borrows the framework of \cite{CS91}, where Carlen and Soffer-- among other results-- derive a stability version of the entropy power inequality by obtaining a monotonicity result along the Ornstein-Uhlenbeck semigroup (in contrast to Stam \cite{Sta59}, whose original proof of the EPI used similar monotonicity result for the heat semigroup), and in the same paper derived CLT convergence even for some dependent variables using entropic methods.  The first EPI stated and proved for potentially dependent random variables, as far as we are aware, comes from Takano \cite{Tak96, Tak98}, who shows that an EPI holds for some pairs of dependent $\mathbb{R}$-valued random variables.  Takano's argument, which also uses monotonicity in heat flow, requires weak dependence as measured by a constant  $\delta_4$, which is the $4$-th moment of the averaged normalized difference between the joint density and the product of the marginals. %being controlled by some function of the Fisher information and variances.  
Based on similar machinery as Takano, a breakthrough in perspective was achieved by Johnson \cite{Joh04:1}, who pursued a {\it conditional} entropy power inequality (i.e., one expressed in terms of conditional entropies) for a pair of real-valued random variables. Rioul provided a more general result in \cite{Rio11}, that unifies both \cite{Joh04:1, Tak96, Tak98}. 

In this paper, we present some new EPIs for dependent summands, with specific attention paid to  demonstrating and exploring ``log-supermodularity'' as a sufficient condition for the veracity of EPIs.   In particular we have the following theorem.

\begin{thm}\label{thm: main from intro}
Let $X, Y$ be $\mathbb{R}^d$-valued random variables such that their joint density is log-supermodular. Then
\begin{align} \label{eq: sup mod EPI}
    e^{\frac{2}{d} h(X+Y)} \geq e^{\frac{2}{d} h(X|Y)} + e^{\frac{2}{d} h(Y|X)}.
\end{align}
\end{thm}
The definition of log-supermodularity will be recalled in subsequent sections, where results will be stated and derived in greater generality.  For the moment we just mention that log-supermodularity can be understood as a specific sort of positive dependence, and that in the presence of negative dependence Theorem \ref{thm: main from intro} is likely to fail.  For example when $X$ and $Y$ are standard normal random variables with negative covariance, \eqref{eq: sup mod EPI} fails. Let us mention that when $X$ and $Y$ are independent, $(X,Y)$ has a log-supermodular density if both $X$ and $Y$ have log-supermodular densities. 
 Thus when $d = 1$, Theorem \ref{thm: main from intro} contains the classical EPI since every density on the real line is trivially log-supermodular.  
 We emphasize that although Theorem \ref{thm: main from intro} will be derived from a more general inequality that does contain the classical EPI, as stated  above  the conditional log-supermodular EPI does not generalize the classical one for $d \geq 2$.

We will also highlight a recently confirmed conjecture (see \cite{MMR25:2}) of Zartash and Robeva \cite{ZR22}, that log-super\-modularity is stable under standard Gaussian convolution.  This allowed Theorem \ref{eq: sup mod EPI} to be stated for random variables with log-supermodular joint density, while the original draft of this article formulated a conditional EPI for variables whose joint density was log-supermodular after convolution with a standard Gaussian\footnote{Log-supermodularity under convolution with a standard Gaussian implies log-supermodularity (at least almost everywhere).  Indeed, by a homogeneity argument, $(X,Y)$ log-supermodular under convolution with a standard Gaussian $Z$ implies that $(X,Y)$ is log-supermodular under convolution with $tZ$ for any $t >0$, taking $t \to 0$ will show that $(X,Y)$ are log-supermodular.}.

Let us outline the rest of the paper.  In Section \ref{sec:notn} we fix notations and definitions requisite for the analysis.  In Section \ref{sec:stam}, we extend Johnson's \cite{Joh04:1} (see also \cite{Tak96}) Fisher information inequality for pairs of random variables to arbitrary number of summands.  In Section \ref{sec:OU-epi}, we derive a general EPI for dependent random variables via the Ornstein-Uhlenbeck flow, with a quantitative error term given by the integral of a Fisher information quantity along the flow.   In Section \ref{sec:cond-EPI} we derive corollaries of the main result, and explain the usefulness of log-supermodularity.  %We conclude the paper in Section \ref{sec:counter} by correcting the record and discussing some erroneous statements in \cite{Joh04:1}; in particular, we give there a counterexample to \cite[Theorem 6]{Joh04:1}.
%, while the second gives some concluding remarks dedicated to the connection of this paper with other recent literature on EPIs. 

\section{Notation and definitions}
\label{sec:notn}

We use the following notation for convenience throughout the remainder of this paper. Let $X_1,\dots,X_n$ each be $\mathbb{R}^d$-valued random variables on a given probability space (we use the term random vector and random variable interchangeably if the dimension is understood from context). Denote by  $X=(X_1,\dots,X_n)$ the $(\mathbb{R}^d)^n$-valued random variable, with components $X_i \in \mathbb{R}^d$ and by $p \colon (\mathbb{R}^d)^n \to \mathbb{R}^+$ the joint density of $X$.  We always assume that $p$ is smooth enough for the relevant derivatives to be well defined in the classical sense.  Note that this will always be the case for densities convolved with Gaussians (for example, along the heat or Ornstein-Uhlenbeck flow).  We also assume that $\mathbb{E}\|X_i\|^2 < \infty$ and that any mentioned density function $f$ is such that $f \log f$
%\in L_1$. 
is integrable (so that the entropy $h(f)$ is well defined and finite). We do not pursue minimal regularity hypotheses here.  In particular and in light of \cite{gavalakis2025conditions}, one might ask if the second moment conditions can be removed.  

We often use the shorthand notation $x_i^j:=(x_i,x_{i+1},\dots,x_j)$, $i,j \in \{1,\dots,n\}$, $i \leq j$. Also, $p$ is function of $x=(x_1,\dots,x_n)$, with, for $i = 1,\dots,n$, $x_i=(x_{i,1},\dots,x_{i,d}) \in \mathbb{R}^d$.

Then, %if we set $V:=- \log p$, 
the score function of the vector $X=(X_1,\dots, X_n)$ is defined for $x \in (\mathbb{R}^d)^n$ as
$$
\rho(x):= -\frac{\nabla p}{p}(x) = (\rho_1(x),\dots,\rho_n(x)) \in (\mathbb{R}^d)^n,
$$
where $\nabla$ denotes the gradient vector.
For simplicity of notations set $\nabla_i=(\frac{\partial}{\partial x_{i,1}},\dots,\frac{\partial}{\partial x_{i,d}})$
so that 
$$
\rho_i = - \frac{\nabla_i p}{p} \in \mathbb{R}^d .
$$
Next we define the Fisher information matrix ${\bf I}=(I_{ij})_{i,j}$ on $\mathbb{R}^n \times \mathbb{R}^n$ as
%$I := \mathbb{E}(\rho \rho^T)$, where $\rho^T$ is the vector $\rho$ transposed, and the expectation has to be understood coefficient by coefficient. In other words,
$$
I_{ij} 
=
I_{ij}(X) \coloneqq
\mathbb{E}_X  
\langle \rho_i, \rho_j \rangle = \mathbb{E} \langle \rho_i(X), \rho_j(X) \rangle
%\frac{1}{p^2} \frac{\partial p}{\partial x_{i,k}} \frac{\partial p}{\partial x_{j,l}} 
 %\qquad k,l=1,\dots,d .
$$
for $i,j =1, \dots,n,$ where $\langle\cdot, \cdot \rangle$ stands for the usual scalar product on $\mathbb{R}^d$.

Now denote by $f \colon \mathbb{R}^d \to \mathbb{R}^+$ the density of the sum $W:=X_1+\dots+X_n$ and define similarly its score function as
$$
\rho_W:= - \frac{\nabla f}{f} \in \mathbb{R}^d .
$$
It will be useful to relate the score function of $W$ to that of $X$. For simplicity of notations set 
%$\nabla_i=(\frac{\partial}{\partial x_{i,1}},\dots,\frac{\partial}{\partial x_{i,d}})$, 
$\bar{x}^{i}$ for the vector $x$ without the coordinates of $x_{i}$, and the corresponding differential form
$d\bar{x}^{i}:=dx_1 \dots dx_{i-1} dx_{i+1}\dots dx_n$.
Then, following Takano \cite{Tak96} we observe that, with the notation $x_w^{(i)} = (x_1^{i-1},w- \sum_{j \neq i} x_j, x_{i+1}^n)$ for any $i=1,\dots,n$,
\begin{align}
\rho_W(w) 
& = 
- \int \dots \int  \frac{\nabla_i p(x_w^{(i)})}{f(w)} d\bar{x}^i \nonumber \\
& =
- \int \dots \int  \frac{\nabla_i p(x_w^{(i)})}{p(x_w^{(i)})} \frac{p(x_w^{(i)})}{f(w)} d\bar{x}^i \nonumber  \\
& =
\mathbb{E}\left(\rho_i(X) | W=w  \right) . \label{contioning}
\end{align}
In particular, for any $(\lambda_1,\dots,\lambda_n) \in \mathbb{R}^n$, it holds
\begin{equation} \label{start}
\left(\sum_{i=1}^n \lambda_i  \right) \rho_W(w) = \mathbb{E}\left(\sum_{i=1}^n \lambda_i \rho_i(X) \bigg| W=w  \right), \qquad w \in \mathbb{R}^d .
\end{equation}
In general, for a random variable $X$ with density $p$, its Fisher information is
$$
I(X) = \int \frac{|\nabla p|^2}{p} .
$$
%if $\langle \cdot, \cdot\rangle$ stands for the usual scalar product of $\mathbb{R}^d$,
Also, we denote by $J$ (or $I(W)=I(X_1+\dots+X_n)$) the Fisher information of $W$, \textit{i.e.}\ 
$$
J \coloneqq \mathbb{E}_W\left( \langle \rho_W,\rho_W \rangle \right) = \mathbb{E} | \rho_W(W) |^2.
$$

Finally we recall the definition of the conditional entropy. Given $X$ an $\mathbb{R}^n$-valued random variable and $Y$ an $\mathbb{R}^k$-valued random variable with a joint density, the conditional entropy of $X$ given $Y$ is defined as
\[
    h(X|Y) \coloneqq h(X,Y) - h(Y).
    %= \int_{\mathbb{R}^{n+k}} p(x,y) \left[ \log \left(\int_{\mathbb{R}^n} p(u,y) du \right) - \log p(x,y) \right] dx dy.
\]

\section{Fisher information inequality for dependent random variables.}
\label{sec:stam}

In this section, we extend the Fisher information inequality for dependent variables obtained by Johnson \cite{Joh04:1} in the case $n=2$ to arbitrary $n$.
%This is a straight forward generalization of  \cite{Joh04:1}. 
Denote by $\mathbf{1}:=(1,\dots,1)$ the vector of $\mathbb{R}^n$ with all ones, and recall that $\langle \cdot, \cdot \rangle$ denotes the scalar product (of $\mathbb{R}^d$ or $\mathbb{R}^n$, depending on the context).

\begin{prop}\label{fisher}
For $\mathbb{R}^d$-valued random vectors $X_1$, $\dots, X_n$ with sufficiently smooth joint density, the inequality
\begin{equation} \label{eq:lambdaprop}
\left(\sum_{i=1}^n \lambda_i \right)^2 J \leq \sum_{i,j=1}^n \lambda_i \lambda_j I_{ij} 
\end{equation}
holds for any $(\lambda_1,\dots,\lambda_n)\in \mathbb{R}^n$.
Furthermore, when the matrix ${\bf I}$ is invertible,
\begin{equation} \label{eq: optimized blachmanstam}
\frac{1}{J} \geq \langle\mathbf{1}, {\bf I}^{-1} \mathbf{1} \rangle .
\end{equation}
%when ${\bf I}^{-1}$, the inverse matrix of ${\bf I}$ exists.%{\color{red}James: Do we want to do anything else here?}
\end{prop}

\begin{proof}
Fix $(\lambda_1,\dots,\lambda_n)\in \mathbb{R}^n$. Since 
$$
\left\| \left(\sum_{i=1}^n \lambda_i \right) \rho_W(W) - \sum_{i=1}^n \lambda_i \rho_i(X)  \right\|^2 \geq 0
$$
expanding the product and taking the expectation, we have that
$$
\left(\sum_{i=1}^n \lambda_i \right)^2 J 
+ 
\sum_{i,j=1}^n \lambda_i \lambda_j I_{ij} 
-
2 \sum_{i=1}^n \lambda_i \sum_{j=1}^n \lambda_j \mathbb{E} \left( \langle \rho_W , \rho_j(X) \rangle\right),
$$
is non-negative.
But \eqref{contioning} guarantees that 
\begin{align*}
\mathbb{E} \left( \langle \rho_W ,\rho_j(X)\rangle \right)
    &= 
        \mathbb{E} \left(  \langle \rho_W ,\mathbb{E} \left( \rho_j(X) | W \right) \rangle\right) 
            \\
    &= 
        \mathbb{E} \left( \langle \rho_W, \rho_W\rangle \right)
            \\
    &= J.
\end{align*}
Hence, for any $(\lambda_1,\dots,\lambda_n)\in \mathbb{R}^n$
it holds
\begin{equation} \label{lambda}
\left(\sum_{i=1}^n \lambda_i \right)^2 J \leq \sum_{i,j=1}^n \lambda_i \lambda_j I_{ij} 
= \lambda^T {\bf I} \lambda 
\end{equation}
as expected, where the last equality is just a rewriting using the matrix representation (here $\lambda:=(\lambda_1,\dots,\lambda_n)$ and $\lambda^T$ is the vector $\lambda$ transposed).

Our aim is now to optimize \eqref{lambda} over all $(\lambda_1,\dots,\lambda_n)\in \mathbb{R}^n$. Indeed, \eqref{lambda} implies that
$$
\frac{1}{J} 
\geq 
\sup_{\lambda_1,\dots,\lambda_n} 
\frac{\left(\sum_{i=1}^n \lambda_i \right)^2}{\sum_{i,j=1}^n \lambda_i \lambda_j I_{ij}}
= 
\sup_{\lambda_1,\dots,\lambda_n} 
\frac{\langle \mathbf{1},\lambda\rangle^2}{\sum_{i,j=1}^n \lambda_i \lambda_j I_{ij}}
$$
where we recall that $\mathbf{1}:=(1,\dots,1)$.
Since ${\bf I}$ is positive semi-definite it has a square root we denote by $\sqrt{\bf I}$ that is invertible when ${\bf I}$ is. Then, setting $\alpha = \sqrt{\bf I}\lambda$,
the supremum above equals
$$
\sup_{\alpha_1,\dots,\alpha_n} \frac{\langle (\sqrt{\bf I})^{-1}\alpha,\mathbf{1}\rangle^2}{\langle \alpha,\alpha\rangle} .
$$ 
By the Cauchy--Schwarz Inequality and symmetry, 
\begin{align*}
\langle (\sqrt{\bf I})^{-1}\alpha,\mathbf{1}\rangle^2 
& = 
\langle \alpha,(\sqrt{\bf I})^{-1}\mathbf{1}\rangle^2 \\
& \leq \langle \alpha,\alpha\rangle \langle(\sqrt{\bf I})^{-1}\mathbf{1},(\sqrt{\bf I})^{-1}\mathbf{1}\rangle \\
& =
\langle \alpha,\alpha\rangle \langle\mathbf{1}, {\bf I}^{-1} \mathbf{1}\rangle
\end{align*}
which ends the proof using the equality case in the  Cauchy--Schwarz Inequality.
\end{proof}

When $d=1$ and the $X_i$ are assumed independent, \eqref{eq:lambdaprop} yields Stam's Fisher information inequality \cite[Equation 2.9]{Sta59} since $I_{ij} = 0$ for $i \neq j$.  Taking $d =1$ and $n =2$, \eqref{eq: optimized blachmanstam} recovers the inequality of Johnson \cite[Theorem 4]{Joh04:1}, while the additional assumption of independence one recovers Blachman \cite[Equation 3]{Bla65}.  The equivalence between \eqref{eq:lambdaprop} and \eqref{eq: optimized blachmanstam} in the independent case was attributed to De Bruijn by Stam in a footnote of \cite{Sta59}.
% \begin{remark}
% We observe that $I$ is symmetric by construction. {\color{red} Cyril: justify that is positive definite, in other words justify that $I^{-1}$ exists}.   
% \end{remark}

% For now on $k=2$ for simplicity. We then write $(X,Y)$ or $(X_1,X_2)$ for the random vector $X_1^2$ of the previous section.
% General $k$ might be considered later on. 

%\section{general background on the Fisher information}

In some situations it will be useful to deal with
the Fisher information %$I(\sum_{i=1}^n \lambda_i X_i)$ 
of the weighted sum $\sum_{i=1}^n \lambda_i X_i$ instead of that of $W=\sum_{i=1}^n X_i$.  
We claim that 
\begin{equation} \label{eq:blachman-lambda}
 I \left(\sum_{i=1}^n \lambda_iX_i \right) \leq 
\sum_{i,j =1}^n  \lambda_i \lambda_j
I_{ij}(X) 
\end{equation}
for any $\lambda_1, \dots,\lambda_n \in (0,1)$ with $\sum_{i=1}^n \lambda_i^2=1$, where as usual $X=(X_1,\dots,X_n)$.
%, and $a_1,\dots,a_n \in \mathbb{R}$, $\sum_{i=1}^n a_i \neq 0$.

To see this, observe first that if $p$ is the joint density of $X=(X_1,\dots,X_n)$, then for $x=(x_1,\dots,x_n) \in (\mathbb{R}^d)^n$
$$
q(x)=\frac{1}{\prod_{i=1}^n \lambda_i^d}p \left(\frac{x_1}{\lambda_1} , \dots,\frac{x_n}{\lambda_n} \right), 
$$
is the density of $(\lambda_1 X_1,\dots,\lambda_nX_n)$. Therefore, changing variables, and using the notation $\frac{x}{\lambda} \coloneqq (\frac{x_1}{\lambda_1},\dots,\frac{x_n}{\lambda_n})$
\begin{align} \nonumber
  &I_{ij}(\lambda_1X_1,\dots,\lambda_n X_n) \\
    &=
  \iint \frac{1}{q(x)} \langle \nabla_i q(x),\nabla_j q(x)\rangle  dx_1\dots dx_n \nonumber \\
&  =
  \iint \frac{\frac{1}{\lambda_i \lambda_j}
  \langle \nabla_i p(\frac{x}{\lambda}) ,\nabla_jp(\frac{x}{\lambda})\rangle/(\lambda_1^d \dots \lambda_n^d)^2}{ p\left(\frac{x}{\lambda}\right)/\lambda_1^d \dots \lambda_n^d} dx_1\dots dx_n \nonumber \\
  & =
  \frac{1}{\lambda_i \lambda_j}   \iint \frac{1}{p(x)} \langle \nabla_i p(x), \nabla_jp(x)\rangle dx_1\dots dx_n \nonumber \\
  & =
    \frac{1}{\lambda_i \lambda_j}  I_{ij}(X) . \label{eq:I-lambda}
\end{align}
By \eqref{eq:lambdaprop} for $t_i$ such that $\sum_{i=1}^n t_i = 1$ and $Y = (Y_1, \dots, Y_n)$,
\[
    I(Y_1 + \cdots + Y_n) \leq \sum_{i,j} t_i t_j I_{ij}(Y).
\]
Applying this to $\lambda$ such that $\sum_{i=1}^n \lambda_i^2 = 1$ and $X = (X_1, \dots, X_n)$ for $Y_i = \lambda_i X$ and $t_i = \lambda_i^2$, gives
\begin{align*}
    I(\lambda_1 X_1 + \cdots + \lambda_n X_n) 
        &\leq 
            \sum_{i,j} \lambda_i^2 \lambda_j^2 I_{ij}(\lambda_1 X_1, \dots, \lambda_n X_n)
                \\
        &= 
            \sum_{i,j} \lambda_i \lambda_j I_{ij}(X),
\end{align*}
where the equality follows from \eqref{eq:I-lambda}. We note that under the assumption of independence  \eqref{eq:blachman-lambda} yields \cite[Lemma 1.3]{CS91}, which also appears as \cite[Theorem 7]{Car91}.

We end this section with a simple but useful remark related to the Fisher information of a single random variable $X_i$ in comparison to that of $X=(X_1,\dots,X_n)$. Given the joint density $p$ of
$X$ (recall that $d\bar{x}^{i}:=dx_1 \dots dx_{i-1} dx_{i+1}\dots dx_n$), the $i$-th marginal is  
$$
p_i(x_i) := \int p(x_1,\dots,x_n)d\bar{x}_i , \qquad x_i \in \mathbb{R}^d 
$$
so that the Fisher information of $X_i$ is given by 
$$
I(X_i) = \int_{\mathbb{R}^d} \frac{|\nabla p_i|^2}{p_i} .
$$
Now, by the Cauchy--Schwarz Inequality,
\begin{align}  \label{eq:comparison}
I(X_i) 
& =
\int_{\mathbb{R}^d} \frac{\left| \int \nabla_i  p(x)d\bar{x}_i \right|^2}{\int   p(x)d\bar{x}_i} dx_i \nonumber \\
& =
\int_{\mathbb{R}^d}  \frac{\left| \int \frac{\nabla_i  p(x)}{\sqrt{p(x)}}\sqrt{p(x)} d\bar{x}_i \right|^2 }{\int   p(x)d\bar{x}_i} dx_i \nonumber \\
& \leq 
\iint_{\mathbb{R}^d \times (\mathbb{R}^d)^{(n-1)}}  \frac{\left|\nabla_i p(x)\right|^2}{ p(x)} d\bar{x}_i dx_i   \nonumber \\
& 
= I_{ii}(X).
\end{align}
Similarly (details are left to the reader), for any $i=1,\dots,n$,
\begin{align} \label{eq:comparison2}
I(X_j, j \neq i)  
   & =
   \int_{(\mathbb{R}^d)^{(n-1)}} \frac{\sum_{j \neq i}  \left| \int \nabla_j  p(x) dx_i \right|^2}{\int   p(x)d\bar{x}_i} d\bar{x}_i \nonumber \\
 & \leq 
 \sum_{j \neq i} I_{jj}(X) .
\end{align}

\section{EPI for dependent variables via the  Ornstein-Uhlenbeck flow}
\label{sec:OU-epi}

In this section we will use the Ornstein-Uhlenbeck approach of Carlen and Soffer \cite{CS91} to the entropy power inequality to get some
 general EPI valid for non independent variables.

Given a density $f \colon \mathbb{R}^m \to \mathbb{R}_+$ define the adjoint (for the Lebesgue measure) operator of the Orstein-Uhlenbeck semi-group $P_t^*$, as the solution of the following PDE in $\mathbb{R}^m$
$$
\frac{d}{dt}P_t^* f = (\Delta +x \cdot \nabla)P_t^*f + mP_t^*f 
$$
where the dot sign stands for the scalar product, $\Delta$ is the Laplacian, $\nabla$ the gradient operator (a vector) and initial condition $P_t^* f|_{t = 0} =f$. This semi-group has an explicit representation formula (Mehler-type formula). For $X$ with density $f$ and $Z \sim N(0,I)$, it holds
\begin{align*}
P_t^*f(x)
& =
e^{mt} \mathbb{E} [f(e^tx -\sqrt{e^{2t}-1}Z)] \\
& =
\mathbb{E} [g_{\sqrt{1-e^{-2t}}}(x-e^{-t}X)] 
\end{align*}
with $g_s(x)=\frac{e^{-|x|^2/2s}}{(2\pi s)^{m/2}}$ the $m$-dimensional centered Gaussian density of covariance $s\mathrm{Id}_{\mathbb{R}^m}$. 
Furthermore, $e^{-t}X + \sqrt{1-e^{-2t}}Z$
(with $X$ and $Z \sim N(0,I)$ independent) has density $P_t^*f$.

One key property is the following relation between the Shannon entropy and the Fisher information (see \cite[Lemma 1.2]{CS91}):
\begin{align}\label{eq:FisherOU}
 h(e^{-t}X &+ \sqrt{1-e^{-2t}}Z) \nonumber
    \\
    &= h(X) + \int_0^t I(e^{-s}X + \sqrt{1-e^{-2s}}Z) ds - mt. 
\end{align}
We are in position to prove the following proposition which constitutes a linearized form of the entropy power inequality for general random variables.

\begin{prop} \label{prop:dependent}
Let $X_1,\dots,X_n$ be $\mathbb{R}^d$-valued random variables and $\tilde{X}_i$ be independent random variables such that $X_i$ {\color{black} and $\tilde{X}_i$ are identically distributed}.  Let $\lambda_1,\dots,\lambda_n \in (0,1)$ so that $\sum_{i=1}^n \lambda_i^2=1$. Set
$Y_i =Y_i(t) := e^{-t}X_i + \sqrt{1-e^{-2t}}Z_i$,
$i=1,\dots,n$, where $Z_i \sim N(0,I)$ are independent (of one another and all $X_i$ and $\tilde{X}_i$) standard Gaussian on $\mathbb{R}^d$, and
$\tilde{Y}_i = \tilde{Y}_i(t) := e^{-t}\tilde{X}_i + \sqrt{1-e^{-2t}}Z_i$,
$i=1,\dots,n$,.
{\color{black} 
Then for all $t \geq 0$ it holds
\begin{align*}
h \left( \sum_{i=1}^n \lambda_i Y_i(t) \right)& - \sum_{i=1}^n \lambda_i^2 h(Y_i(t)) 
\\
&\leq 
h\left( \sum_{i=1}^n \lambda_i  X_i\right) - \sum_{i=1}^n \lambda_i^2 h(X_1) 
+ \lambda^T {\bf R}_t \lambda
\end{align*}
with
\begin{align*}
{\bf R}_t \coloneqq  \int_0^t
 \left(
 %{\bf I}(Y(s)) - {\bf I}(\tilde{Y}(s) 
{\bf I}(s) - \tilde{\bf I}(s)  
 \right)  ds
\end{align*}
}
{\color{black}
where ${\bf I}(s):={\bf I}(Y(s))$ and $\tilde{\bf I}(s):={\bf I}(\tilde{Y}(s))$ are the Fisher information matrices of the random vector 
$Y(s)=(Y_1(s),\dots, Y_n(s))$ 
 and
$\tilde{Y}(s)=(\tilde{Y}_1(s),\dots, \tilde{Y}_n(s))$ respectively.
}
In particular,
$$
h\left( \sum_{i=1}^n \lambda_i  X_i \right) 
\geq  
\sum_{i=1}^n \lambda_i^2 h(X_i) 
- \limsup_{t \to \infty} \lambda^T {\bf R}_t \lambda.
$$
\end{prop}
%Above (and below) it is understood that $\tilde X_i$ depends on $s$ or $t$ according to the context.

\begin{remark}
Note that for independent random variables $X_1,\dots,X_n$, $R_t(X_1,\dots,X_n)=0$ for all $t$, therefore recovering Carlen-Soffer's inequality \cite[Theorem 1.1]{CS91} and Stam's linearized version of the entropy power inequality.
\end{remark}

\begin{proof}
Fix $t >0$. Applying \eqref{eq:blachman-lambda} to $Y_i(s)= e^{-s}X_i + \sqrt{1-e^{-2s}}Z_i$ with $s \leq t$, we have
$$
I\left(\sum_{i=1}^n \lambda_iY_i(s) \right) \leq \lambda^t {\bf I}(Y(s)) \lambda.
$$
% $$
% I(\lambda_1 \tilde X_1+\lambda_2 \tilde X_2) \leq 
%  \lambda_1^2 I_{11}(\tilde X_1,\tilde X_2)
% +
% \lambda_2^2 I_{22}(\tilde X_1,\tilde X_2)
% +
% 2\lambda_1\lambda_2  I_{12}(\tilde X_1,\tilde X_2) .
% $$
Applying the the fundamental theorem of calculus term wise to the summation
\begin{align*}
    \sum_{i=1}^n \lambda_i^2 h(Y_i(t)) &- \sum_{i=1}^n \lambda_i^2 h(X_i)
            \\
        &=
            \sum_{i=1} \lambda_i^2 \left( \int_0^t I(Y_i(s)) ds  - t d \right)
                \\
        &=
            \lambda^T \left(\int_0^t {\bf I}( \tilde{Y}(s) ) ds\right) \ \lambda - td.
\end{align*}
By the same logic, and applying \eqref{eq:FisherOU} we have
\begin{align*}
h &\left( \sum_{i=1}^n \lambda_i Y_i(t) \right)  - h \left( \sum_{i=1}^n \lambda_i X_i \right)
\\
    &=
 \int_0^t I \left( \sum_{i=1}^n \lambda_i Y_i(s) \lambda \right) ds -  t \ d\\
 & \leq 
\lambda^T \int_0^t 
\left( 
{\bf I}(Y(s)) \right) ds \lambda -td.
%  & =
%  \int_0^t 
% \left( 
% \lambda_1^2 I(\tilde X_1)
% +
% \lambda_2^2 I(\tilde X_2)
%  \right) ds - t
%  +R_t(X_1,X_2) \\
% & =
% \sum_{i=1}^n \lambda_i^2 \left( \int_0^t  I(Y_i(s)) ds
% -dt \right)
%  + R_t  \\
%  & =
% \sum_{i=1}^n \lambda_i^2 \left( h(Y_i(t)) - h(X_i) \right)
%  +R_t.
\end{align*}
Combining this with the previous equality leads to the first part of the proposition.

For the second part it is enough to take the limit $t \to \infty$ and to observe that
$$
h\left( \sum_{i=1}^n\lambda_i Y_i(t) \right) - \sum_{i=1}^n \lambda_i^2 h(Y_i(t))
$$ 
converges to
$h(Z_{n+1}) - \sum_{i=1}^n \lambda_i^2 h(Z_i) $ where $Z_i \sim N(0,I)$ are independent. Since $h(Z_i)=\frac{d}{2}\log(2\pi e)$ and $\sum \lambda_i^2=1$, $h(Z_{n+1}) - \sum_{i=1}^n \lambda_i^2 h(Z_i)  = 0$, ending the proof of the proposition.
\end{proof}
% {\color{red}
% \begin{coro}
%     For $X = (X_1, \dots, X_n)$ with log-supermodular density $f$ and $\lambda_i \geq 0$ with $\sum_{i=1}^n \lambda_i^2 = 1$,
%     \[
%         h \left( \sum_{i=1}^n \lambda_i X_i \right) \geq \sum_{i=1}^n \lambda_i^2 h(X_i).
%     \]
% \end{coro}

% {\color{black}Cyril: the log-supermodularity property has not yet been defined. Either we give in the introduction the definition of lsm (and few properties), or we move this Corollary after having defined lsm densities?}

% \begin{proof}
%     By \cite{MMR25:2}, with $Y(t), \tilde{Y}(t),$ and ${\bf R}_t$ as defined as in Proposition \ref{prop:dependent}.  Writing $f \coloneqq f(t,x)$ the (necessarily smooth as a Gaussian convolution) density of $Y(t)$. By log-supermodularity (and the smootheness of $f$), for $i \neq j$
%     \[
%         \frac{\partial}{\partial x_i} \frac{\partial}{\partial x_j} \log f \geq 0
%     \]
%     Thus for $i \neq j$ by integration by parts
%     \[
%         0 \leq \int f \frac{d}{dx_i} \frac{d}{dx_j} \log f = - \int \frac{(\frac{\partial f}{\partial x_i} )(\frac{\partial f}{\partial x_i} )}{f} = - I_{ij}(Y(t)).
%     \]
%     Thus it follows that for all $t$, the matrix ${\bf I}(Y(t)) - {\bf I}(\tilde{Y}(t))$ has all non-positive off diagonal entries, and since the matrices ${\bf I}(Y(t))$ and ${\bf I}(\tilde{Y}(t))$ agree along the diagonal, ${\bf R}_t$ does as well has all non-positive entries.  Hence $\limsup_{t \to \infty} \lambda^T {\bf R}_t \lambda \leq 0$ and the result follows.
% \end{proof}
% }

As a direct consequence, we obtain the following entropy power inequality for general random variables.

\begin{coro} \label{cor:dependent}
Let $X_1,\dots,X_n$ be $\mathbb{R}^d$-valued random variables and, for $i=1,\dots, n$, 
$$
\lambda_i^2 := \frac{e^{\frac{2}{d} h(X_i)}}{e^{\frac{2}{d} h(X)}+\dots+e^{\frac{2}{d} h(X_n)}}
% ,
% \qquad
% \lambda_Y^2 := \frac{e^{2h(Y)}}{e^{2h(X)}+e^{2h(Y)}} .
$$
Set $\bar X_i=X_i / \lambda_i$
and  
$\bar Y_i(s) = e^{-s}\bar X_i + \sqrt{1-e^{-2s}}Z_i$, $s \geq 0$,
%$\tilde Y = e^{-t}Y + \sqrt{1-e^{-2t}}Z_Y$
where $Z_1, \dots,Z_n \sim N(0,I)$ are independent standard Gaussian (on $\mathbb{R}^d$). Finally denote
\begin{align*}
\bar R 
    &= 
        \bar R(X_1,\dots,X_n)
            \\
%    &= 
%        \frac{2}{d} \int_0^\infty
%\left( \sum_{i=1}^n \lambda_i^2 (\bar I_{ii}(s) - %I( \bar Y_i(s))) 
%+
%\sum_{\genfrac{}{}{0pt}{}{i,j=1}{i \neq j}}^n %\lambda_i \lambda_j \bar I_{ij}(s) 
%\right) ds \\
& {\color{black} =
\frac{2}{d} 
\lambda^T \left( \int_0^\infty \left[\bar{\bf  I}(s) - \mathrm{diag}(I(\bar Y_i(s)))_i \right]   ds \right) \lambda }
\end{align*}
{\color{black}where $\mathrm{diag}(I(Y_i(s)))_i )$ is the ($n \times n$) diagonal matrix with diagonal entries $I( \bar Y_i(s)))$, $i=1,\dots,n$ and
$\bar{\bf I}(s) \coloneqq {\bf I}(\bar{Y}(s))$
is the Fisher information matrix of the random vector $\bar Y(s)=(\bar Y_1(s),\dots,\bar Y_n(s))$.
}
%where we set for simplicity $\bar I_{ij}(s):=I_{ij}(\bar Y_1(s),\dots,\bar Y_n(t))$.
Then,
$$
e^{\frac{2}{d} h(X_1+\dots+X_n)} 
\geq  
\left(\sum_{i=1}^n e^{\frac{2}{d} h(X_i)}  \right) e^{- \bar R} .
$$
\end{coro}

%Understanding log-supermodularity as a positive correlation inequality, this says that the EPI continues to hold under this notion of ``positive'' or non-negative correlation.

\begin{remark}
As for the previous proposition, if $X_i$  are independent $\bar R(X_1,\dots,X_n)=0$ and we recover the usual entropy power inequality.

%{\color{red}Cyril: In general, as already observed in \eqref{eq:comparison},
%$I(Y_i(s)) \leq I_{ii}s)$. %Therefore $R$ looks positive a priori. However there might be some situation where $R \leq 0$ thanks to the cross terms $I_{ij}$.  to be investigated maybe.}
\end{remark}

\begin{proof}
It follows by Proposition \ref{prop:dependent}
applied to $\bar X_i=X_i/\lambda_i$, $i=1,\dots,n$  that
\begin{align*}
&e^{\frac{2}{d} h(X_1+\dots+X_n)} 
 = 
e^{\frac{2}{d} h(\lambda_1 \bar X_1 + \dots + \lambda_n \bar X_n)} \\
& \geq 
\exp \left\{ \frac{2 \lambda_1^2 }{d}h( \bar X_1) + \dots + \frac{2 \lambda_n^2}{d}h( \bar X_n) \right\} e^{-R(\bar X_1,\dots,\bar X_n)} \\
& =
\prod_{i=1}^n \left( \frac{e^{\frac{2}{d} h(X_i)}}{\lambda_i^2}\right)^{\lambda_i^2} 
%\left( \frac{e^{2h(Y)}}{\lambda_2^2}\right)^{\lambda_2^2} 
e^{-\bar R(X_1,\dots,X_n)}
\end{align*}
since $h(\bar X_i) = h(X_i/\lambda_i)=h(X_i)-d \log \lambda_i$ and $R(\bar X_1,\dots,\bar X_n)=\bar R(X_1,\dots,X_n)$. The expected result follows.
\end{proof}

\section{Entropy power inequality for conditional entropy}
\label{sec:cond-EPI}

Our next aim is to deal with entropy power inequality for conditional entropy. We start with a linearized form similar to Proposition \ref{prop:dependent}.

\begin{prop}\label{prop:conditionalEPI}
Let $X_1,\dots,X_n$ be $\mathbb{R}^d$-valued random variables and $\lambda_1,\dots,\lambda_n \in (0,1)$ so that $\sum_{i=1}^n \lambda_i^2=1$. Set
$Y_i(s) = e^{-s}X_i + \sqrt{1-e^{-2s}}Z_i$, $s \geq 0$, $i=1,\dots,n$, where $Z_1,\dots,Z_n \sim N(0,I)$ are independent standard Gaussian (on $\mathbb{R}^d$).
Then it holds
\begin{align*}
h &\left(\sum_{i=1}^n \lambda_i Y_i  \right) 
 - \sum_{i=1}^n \lambda_i^2 h(Y_i|Y_j, j \neq i)
 \\
&\leq 
h \left( \sum_{i=1}^n \lambda_i  X_i \right) - \sum_{i=1}^n \lambda_i^2 h(X_i|X_j, j \neq i) 
+ S_t
\end{align*}
where $$S_t = S_t(X_1,\dots,X_n)= \int_0^t \mathcal{S}_s ds 
$$
for
\begin{align*} \mathcal{S}_s \coloneqq&
\sum_{i,j=1}^n \lambda_i \lambda_j I_{ij}(Y_1(s),\dots,Y_n(s))  \\
&- \sum_{i=1}^n I_{ii}(Y_1(s),\dots,Y_n(s))
+ \sum_{i=1}^n \lambda_i^2 I(Y_j(s), j\neq i) .
\end{align*}
% \left( \lambda_1^2(I_{11}(\tilde X_1,\tilde X_2) - I( \tilde X_1)) 
% +
% \lambda_2^2(I_{22}(\tilde X_1,\tilde X_2) - I( \tilde X_2))
% +
% 2 \lambda_1 \lambda_2 I_{12}(\tilde X_1,\tilde X_2) 
% \right) ds . 
In particular,
$$
h \left( \sum_{i=1}^n \lambda_i  X_i \right) \geq  \sum_{i=1}^n \lambda_i^2 h(X_i|X_j, j \neq i) 
- \limsup_{t \to \infty} S_t .
$$
\end{prop}

\begin{proof}
Observe that, for any $i=1,\dots,n$,
\begin{equation} \label{eq:conditionalH}
h(Y_i(t)|Y_j(t), j \neq i) = 
h(Y_1(t),\dots,Y_n(t)) - h(Y_j(t), j \neq i) .
\end{equation}
Therefore, 
%thanks to \eqref{eq:conditionalH} (applied twice),
\begin{align*}
 h\left( \sum_{i=1}^n \lambda_i Y_i(t) \right) 
 &- \sum_{i=1}^n \lambda_i^2 h(Y_i(t)|Y_j(t), j \neq i)
 \\
& = 
 h\left( \sum_{i=1}^n \lambda_i Y_i(t) \right) 
 - 
 h(Y_1(t),\dots,Y_n(t))+ \sum_{i=1}^n \lambda_i^2 h( Y_j(t), j\neq i )  .
\end{align*}
Hence, by \eqref{eq:FisherOU}, it holds
\begin{align*}
 h&\left( \sum_{i=1}^n \lambda_i Y_i(t) \right) 
 - \sum_{i=1}^n \lambda_i^2 h(Y_i(t)|Y_j(t), j \neq i) \\
& =
 h\left( \sum_{i=1}^n \lambda_i X_i \right) 
 - h(X_1,\dots,X_n)
 + \sum_{i=1}^n \lambda_i^2 h(X_j, j \neq i)
 \\
&\quad  + \int_0^t \bigg[I \left( \sum_{i=1}^n \lambda_i Y_i(s) \right) - I(Y_1(s), \dots,Y_n(s))
 + \sum_{i=1}^n \lambda_i^2 I(Y_j(s), j \neq i) \bigg] ds \\
& =
 h\left( \sum_{i=1}^n \lambda_i X_i \right) 
 - \sum_{i=1}^n \lambda_i^2 h(X_i|X_j, j \neq i)
 \\
&\quad  + \int_0^t \bigg[I \left( \sum_{i=1}^n \lambda_i Y_i(s) \right) - I(Y_1(s), \dots,Y_n(s)) + \sum_{i=1}^n \lambda_i^2 I(Y_j(s), j \neq i) \bigg] ds 
\end{align*}
where in the last equality we used  \eqref{eq:conditionalH} at time $t=0$ to reconstruct the conditional entropies.
Applying the Fisher information inequality \eqref{eq:blachman-lambda}, we obtain
$$
I \left( \sum_{i=1}^n \lambda_i Y_i(s) \right) 
\leq 
\sum_{i,j=1}^n \lambda_{i} \lambda_j I_{ij}(s) 
$$
where we set for simplicity $I_{ij}(s):=I_{ij}(Y_1(s),\dots,Y_n(s))$.
On the other hand, if we denote $p_s$ the density of $(Y_1(s),\dots, Y_n(s))$,
\begin{align*}
 I(Y_1(s),\dots,Y_n(s))
& =
 \int_{(\mathbb{R}^d)^n} \frac{|\nabla p_s|^2}{p_s} 
 =
\sum_{i=1}^n I_{ii}(s) .
\end{align*}
It follows that 
\begin{align*}
& h\left( \sum_{i=1}^n \lambda_i Y_i(t) \right) 
 - \sum_{i=1}^n \lambda_i^2 h(Y_i(t)|Y_j(t), j \neq i)
    \\
&\leq 
h\left( \sum_{i=1}^n \lambda_i X_i\right) 
 - \sum_{i=1}^n \lambda_i^2 h(X_i|X_j, j \neq i) \\
&\quad  + 
\int_0^t \bigg[ \sum_{i,j=1}^n \lambda_i \lambda_j I_{ij}(s)  
- \sum_{i=1}^n I_{ii}(s)
 + \sum_{i=1}^n \lambda_i^2 I(Y_j(s), j\neq i) \bigg] ds 
\end{align*}
as expected.

The second part of the theorem follows from the fact that, in the limit $t \to \infty$,
\begin{align*}
h \left(\sum_{i=1}^n \lambda_i Y_i  \right) 
 - \sum_{i=1}^n \lambda_i^2 h&(Y_i|Y_j, j \neq i)
 \to \\
 &h \left(Z_{n+1} \right) 
 - \sum_{i=1}^n \lambda_i^2 h(Z_i),
 \end{align*}
where $Z_i$, $i=1,\dots,n+1$, are \textit{i.i.d.}\ standard Gaussian variables in $\mathbb{R}^d$, for which it is known that $h(Z_i)=\frac{d}{2}\log(2\pi e)$. In particular 
$ h \left(Z_{n+1} \right) 
 - \sum_{i=1}^n \lambda_i^2 h(Z_i) = 0$ leading to the desired conclusion.
\end{proof}

As a direct consequence, we deduce a general entropy power inequality for conditional entropy.

\begin{coro} \label{cor:conditionalEPI}
Let $X_1,\dots,X_n$ be $\mathbb{R}^d$-valued random variables and, for $i=1,\dots, n$, 
$$
\lambda_i^2 := \frac{e^{\frac{2}{d} h(X_i|X_j, j \neq i)}}{\sum_{i=1}^n e^{\frac{2}{d} h(X_i|X_j, j \neq i)}} .
$$
Set $\bar X_i=X_i / \lambda_i$
and  
$\bar Y_i(s) = e^{-s}\bar X_i + \sqrt{1-e^{-2s}}Z_i$, $s \geq 0$,
%$\tilde Y = e^{-t}Y + \sqrt{1-e^{-2t}}Z_Y$
where $Z_1, \dots,Z_n \sim N(0,I)$ are independent standard Gaussian (on $\mathbb{R}^d$). Finally denote
\begin{align*}
\bar S 
&= 
\bar S(X_1,\dots,X_n) 
    \\
&= 
\frac{2}{d} \int_0^\infty \bigg[ \sum_{i,j=1}^n \lambda_i \lambda_j \bar I_{ij}(s)  
- \sum_{i=1}^n \bar I_{ii}(s) 
 + \sum_{i=1}^n \lambda_i^2 I(\bar Y_j(s), j\neq i) \bigg] ds 
\end{align*}
where we set for simplicity $\bar I_{ij}(s):=I_{ij}(\bar Y_1(s),\dots,\bar Y_n(s))$.
Then it holds
\begin{equation} \label{eq:EPI-conditional}
e^{\frac{2}{d} h(X_1\dots+X_n)} \geq \left( \sum_{i=1}^n e^{\frac{2}{d}  h(X_i|X_j, j \neq i)}\right) e^{-\bar S}    .
\end{equation}
In particular, if $\bar I_{ij}(s) \leq 0$ for all $i \neq j$ and all $s >0$, then
$$
e^{\frac{2}{d} h(X_1\dots+X_n)} \geq  \sum_{i=1}^n e^{\frac{2}{d} h(X_i|X_j, j \neq i)} .
$$
\end{coro}

\begin{proof}
By Proposition \ref{prop:conditionalEPI} applied to $\bar X_i$, it holds 
\begin{align*}
&e^{\frac{2}{d} h(X_1+\dots+X_n)} 
 = 
e^{ \frac{2}{d} h(\lambda_1 \bar X_1 + \dots + \lambda_2 \bar X_2)} \\
& \geq 
\exp \left\{ \frac{2}{d}  \sum_{i=1}^n \lambda_i^2h(\bar X_i|\bar X_j, j \neq i)  \right\} e^{-\limsup_{t \to \infty} S_t(\bar X_1,\dots,\bar X_n)} \\
& =
\prod_{i=1}^n \left( \frac{e^{\frac{2}{d}  h(X_i|X_j, j \neq i)}}{\lambda_i^2}\right)^{\lambda_i^2} 
 e^{-\bar S} 
\end{align*}
since $h(\bar X_i|\bar X_j, j \neq i) 
%= h(X_i/\lambda_i|X_j/\lambda_j, j \neq i)
=h(X_i|X_j, j \neq i)-\log \lambda_i$ for all $i=1,\dots,n$, and 
$$
\limsup_{t \to \infty} S_t(\bar X_1,\dots,\bar X_n)= \bar S(X_1,\dots,X_n) .
$$
The first result follows.

For the second conclusion, by \eqref{eq:comparison2}
applied to $\bar Y_1(s),\dots,\bar Y_n(s)$, for all $i=1,\dots,n$, it holds
\begin{align*}
I(\bar Y_j(s), j \neq i)  
& \leq 
\sum_{j \neq i} \bar I_{jj}(s) \\
& =
\left(\sum_{j=1}^n \bar I_{jj}(s) \right) - \bar I_{ii}(s) 
.
\end{align*}
Therefore, since $\sum_{i=1}^n  \lambda_i^2 = 1$,
$$
\sum_{i = 1}^n \lambda_i^2 I(\bar Y_j, j \neq i)
\leq 
\sum_{i=1}^n \bar I_{ii}(s) - \sum_{i = 1}^n \lambda_i^2 \bar I_{ii}(s) .
$$
In turn,
\begin{align*}
\bar S 
& = 
2\int_0^\infty \bigg[ \sum_{i,j=1}^n \lambda_i \lambda_j \bar I_{ij}(s)  
- \sum_{i=1}^n \bar I_{ii}(s) 
+ \sum_{i=1}^n \lambda_i^2 I(\bar Y_j(s), j\neq i) \bigg] ds     \\
& \leq 
2\int_0^\infty  \sum_{\genfrac{}{}{0pt}{}{i,j=1}{i \neq j}}^n \lambda_i \lambda_j  \bar I_{ij}(s)  
 ds .
\end{align*}
The assumption $\bar I_{ij}(s) \leq 0$ ensures that 
$\bar S \leq 0$ leading to the desired entropy power inequality for conditional entropy.
\end{proof}

\begin{remark}
The reader might be surprised by the opposite expression of $\bar R$ in Corollary \ref{cor:dependent} and  $\bar S$ in Corollary \ref{cor:conditionalEPI}. Indeed, for $n=2$, $\bar R$ and $\bar S$ take the form (after some algebra and using that $\lambda_1^2+\lambda_2^2=1$)
\begin{align*}
\bar R &= 
\frac{d}{2} \int_0^\infty
\bigg[
\lambda_1^2 (\bar I_{11}(s) - I( \bar Y_1(s)))
+\lambda_2^2 (\bar I_{22}(s) - I( \bar Y_2(s)))
 +
2 \lambda_1 \lambda_2 \bar I_{12}(t) 
\bigg] dt  
\end{align*}
while
\begin{align*}
\bar S 
&= 
\frac{d}{2}  \int_0^\infty \bigg[ 
- \lambda_1^2 (\bar I_{22}(s) - I( \bar Y_2(s)))
-\lambda_2^2 (\bar I_{11}(s) - I( \bar Y_1(s)))
 +
2 \lambda_1 \lambda_2 \bar I_{12}(s) 
\bigg] ds.
\end{align*}
We emphasize the terms $\bar I_{ii}(t) - I( \bar Y_i(t))$  appear in both  $\bar R$ and $\bar S$,
but with opposite signs.  This  change arises from the 
decomposition formula of the entropy \eqref{eq:conditionalH}.  The conditional entropy power inequality (Corollary \ref{cor:conditionalEPI}), leverages that  the non-negativity of $\bar I_{ii}(t) - I( \bar Y_i(t))$ (see  \eqref{eq:comparison}) can be combined with the assumption $I_{12}(s) \leq 0$, to directly yield $\bar S \leq 0$.  Such an argument cannot be applied directly to guarantee the negativity of $\bar R$. 

% In \cite{Joh04:1}, it is claimed that, under the condition\footnote{Strictly speaking, in \cite{Joh04:1}, the author is dealing with the heat flow and not the Ornstein-Uhlenbeck flow. However this is just a matter of scaling and reformulation as there is a correspondence between the two flows.} $I_{12}(s) \geq 0$, the following conditional entropy power inequality
% $$
% e^{2h(X_1+X_2)} \geq e^{2h(X_1|X_2)} + e^{2h(X_2|X_1)}
% $$
% holds. This goes in the opposite direction of our condition $I_{12}(s) \leq 0$. In fact, as we will show in the appendix, with a counter-example, the statement in \cite{Joh04:1} is incorrect.
\end{remark}

In what follows we give a sufficient condition on the density $p$ of the random vector $(X_1,\dots,X_n)$ for the condition 
$\bar I_{ij}(s) \leq 0$ in Corollary \ref{cor:conditionalEPI} to hold for all $s >0$.  To that aim we need to introduce the notion of log-supermodular functions on $(\mathbb{R}^d)^n$.

\begin{defn}
A function $u \colon (\mathbb{R}^d)^n \to (0,\infty)$ is said to be log-supermodular if for all $x=(x_1,\dots,x_n), y=(y_1,\dots,y_n) \in (\mathbb{R}^d)^n$ it holds
$$
u(x) u(y) \leq u(x \wedge y) u(x \vee y) 
$$
where 
 $x \vee y \in (\mathbb{R}^d)^n$ denotes the componentwise maximum of $x$ and $y$ and $x \wedge y \in (\mathbb{R}^d)^n$
denotes the componentwise minimum of $x$ and $y$. 
\end{defn}

To be precise, if $x_i=(x_{i,1},\dots,x_{i,d})$, $i=1,\dots,n$, and similarly for $y$,
$$
x \wedge y = (\min(x_{1,1},y_{1,1}), \dots ,\min(x_{n,d},y_{n,d})) 
$$
and
$$
x \vee y = (\max(x_{1,1},y_{1,1}),\dots,\max(x_{n,d},y_{n,d})).
$$

The class of log-supermodular densities is widely studied in various field of mathematics; we refer to the introduction of \cite{ZR22} for an account of the literature and discussion. Note that, in some older literature, log-supermodular densities are called multivariate totally positive of order 2 (MTP$_2$).

It is known (see, e.g., \cite{Top98:book} or \cite[Proposition 2.5]{FMZ24}) that a continuously twice-differentiable function $u$ is log-supermodular if and only if $\frac{\partial^2}{\partial x_{i,k} \partial x_{j,l}} \log u \geq 0$ for all 
distinct $(i,k), (j,l)$, $i,j=1,\dots, n$,  $k,l = 1 \dots,d$.

We will need to deal with log-supermodular densities that remain log-supermodular after convolution with a Gaussian with covariance matrix proportional to the identity matrix (\textit{i.e.}\ of the form $\kappa I$, with $\kappa >0$ and $I$ the identity matrix). We will call this class of densities $\mathcal{C}$. Namely, if $g_s$ stands for the centered Gaussian density with covariance matrix $sI_{nd}$ ($s>0$) on $\mathbb{R}^{nd}$, and let $\mathcal{L}$ denote the space of log-supermodular density functions we denote by
$$
\mathcal{C} \coloneqq \left\{ p \in \mathcal{L} : p * g_s \in \mathcal{L}, \forall
s>0 \right\}.
$$
As discussed in Karlin and Rinott \cite{KR80:2} convolution of any two log-supermodular  densities need not be log-supermodular (see 
\cite[page 486]{KR80:2} for a counterexample with two Gaussian densities). 
In \cite{KR83:1} it is proved that a Gaussian density with covariance matrix $\Sigma$ is log-supermodular if and only if the off-diagonal entries of $\Sigma^{-1}$ are non-positive. In \cite[Theorem 6]{ZR22}
Zartash and Robeva give some conditions for a log-supermodular density $p$ to belong to $\mathcal{C}$. 
Note that by a straightforward scaling argument, having $p*g_s$ log-supermodular for all $s>0$ in the definition of $\mathcal{C}$ is equivalent to having $p*g_{s_o}$ log-supermodular for any specific $s_o>0$.
In particular $\mathcal{C} \neq \emptyset$. Moreover, they conjecture (see \cite[Conjecture 10]{ZR22}) that, in fact, $\mathcal{C}$ coincides with the class of all log-supermodular densities. We will prove in a forthcoming paper \cite{MMR25:2} that their conjecture is true and actually holds in more generality.

We will prove the following Corollary.

\begin{coro}\label{cor:submodular}
Let %$p \in \mathcal{C}$ 
$p$ be a  log-supermodular density of a random vector $X=(X_1,\dots,X_n) \in (\mathbb{R}^d)^n$. 
Then it holds
$$
e^{\frac{2}{d}h(X_1+\dots+X_n)} \geq e^{\frac{2}{d}h(X_1|X_i, i \neq 1)} + \dots +  e^{\frac{2}{d}h(X_n|X_i, i \neq n)} .
$$
\end{coro}

\begin{proof}
We observe first that, if $u \in \mathcal{C}$ is twice differentiable on $(\mathbb{R}^d)^n$, by definition,
\begin{align*}
P_t^*u((x))
& =
\mathbb{E}[g_{\sqrt{1-e^{-2t}}}(x-e^{-t}X)] \\
& = 
\int_{\mathbb{R}^{nd}} u(y) g_{\sqrt{1-e^{-2t}}}(x-e^{-t}y) dy   \\
& =
e^{nt} \int_{\mathbb{R}^{nd}} u(z_1e^{t},\dots,z_ne^{t}) g_{\sqrt{1-e^{-2t}}}(x-z) dz %\quad \mbox{(changing variables)} 
\\
& =
u_t * g_{\sqrt{1-e^{-2t}}} (x),
\end{align*}
where we set $u_t(z):=e^{nt} u(z_1e^{t},\dots,z_ne^{t})$.
Therefore, if $u \in \mathcal{C}$,  $u_t$ is log-supermodular, and by definition of $\mathcal{C}$, the convolution $u_t * g_{\sqrt{1-e^{-2t}}}$
is also log-supermodular. Equivalently this can be rephrased as: if $u \in \mathcal{C}$ is smooth enough,  $P_t^*u$ is log-supermodular. 

Our aim is now to prove that, with the notations of Corollary \ref{cor:conditionalEPI}, it holds
$$
I_{ij}(\bar Y_1(s),\dots,Y_n(s)) \leq 0
$$
for all $i \neq j$, and to apply Corollary \ref{cor:conditionalEPI}.
Since $p$ is the density of $(X_1,\dots,X_n)$, for $(x_1,\dots,x_n) \in (\mathbb{R}^d)^n$,
$$
\bar p(x_1,\dots,x_n) = \lambda_1 \dots \lambda_n p(\lambda_1x_1,\dots,\lambda_n x_n),
$$
is the density of the random vector
$(\bar X_1,\dots,\bar X_n)=(\frac{X_1}{\lambda_1},\dots,\frac{X_n}{\lambda_n})$ (the $\lambda_i$'s are defined in Corollary \ref{cor:conditionalEPI}, they satisfy $\lambda_i \in (0,1)$ and $\sum \lambda_i^2=1$). Denote by $\bar p_s$ the density of
$(\bar Y_1(s),\dots,\bar Y_n(s))$ where we recall that
$\bar Y_i(s) = e^{-s}\bar X_i + \sqrt{1-e^{-2s}}Z_i$,
with $Z_i \sim N(0,I)$ an independent standard Gaussian. We know that 
$\bar p_s = P_s^* \bar p$ and by the above observation that $\bar p_s$ is log-supermodular.

Recall that a smooth density $u$ is log-supermodular if and only if $\frac{\partial^2}{\partial x_{i,k}\partial x_{j,l}}\log u \geq 0$ for all distinct $(i,k),(j,l)$, $i,j=1,\dots,n$, $k,l=1,\dots,d$. Hence,
$\frac{\partial^2}{\partial x_{i,k}\partial x_{j,l}}\log \bar p_s \geq 0$ for all distinct $(i,k),(j,l)$.

Finally, integrating by parts, we get
\begin{align*}
I_{ij}(\bar Y_1(s),\dots,\bar Y_n(s)) 
& = \int_{(\mathbb{R}^d)^n} \frac{\langle \nabla_i \bar p_s,  \nabla_j \bar p_s\rangle}{\bar p_s}  \\
& =
 \sum_{k=1}^d \int_{(\mathbb{R}^d)^n} \frac{\partial \log \bar p_s}{\partial x_{i,k}} \frac{\partial \bar p_s}{\partial x_{j,k}} \\
& =
-   \sum_{k=1}^d \int \frac{\partial^2 \log \bar p_s}{\partial x_{i,k} \partial x_{j,k}}  \bar p_s \leq 0,
\end{align*}
which is the expected result.
This ends the proof of the corollary by means of Corollary \ref{cor:conditionalEPI}.
\end{proof}

%\section{On the conditional entropy power inequality of Johnson}
%\label{sec:counter}

\section{Final remarks and comments on the existing literature}

In this section we make some connections between our results and the existing literature. We mainly focus on a result by Rioul \cite{Rio11} and an other one by Hao and Jog
\cite{HJ18:isit}

\subsection{On the entropy power inequality of Rioul}

In \cite{Rio11} the author derives some general entropy power inequality valid for any random vector, under some assumption involving the Fisher information
of the evolute of the random vector under heat flow.
In particular, the author is considering variables of the type $X_s=X+\sqrt{s}Z$ where $Z$ is an independent standard normal, whose density is well known to be  equal, up to change of time and scale, to the density $P_s^*f$ of $e^{-s}X+\sqrt{1-e^{-2s}}Z$. In particular this changes of time and scale does not affect the positivity of Fisher information matrices.

Now the assumption of Rioul (see \cite[Section V, Corollary 2]{Rio11})
reads, for $d=1$, in our language and under the notations of Corollary \ref{cor:dependent}, 
$\bar{\bf  I}(s) - \mathrm{diag}(I(\bar Y_i(s)))_i$
negative semi-definite. In turn, under this assumption, Corollary \ref{cor:dependent} shows that 
$$
e^{\frac{2}{d} h(X_1+\dots+X_n)} 
\geq  
\left(\sum_{i=1}^n e^{\frac{2}{d} h(X_i)}  \right)  
$$
therefore recovering his entropy power inequality for dependent random variables. In that sense, for $d=1$, our result (Corollary \ref{cor:dependent}) can be seen as a quantitative version of Rioul's entropy power inequality.

We also mention that, as stated in \cite[Section V, Corollary 3]{Rio11}, the negative semi-definiteness of $\bar{\bf  I}(s) - \mathrm{diag}(I(\bar Y_i(s)))_i$ implies earlier conditions proposed by Takano \cite{Tak96} and Johnson \cite{Joh04:1} to guarantee an entropy power inequality to hold for dependent random variables.

\subsection{Relation with Hao-Jog's inequality}

%Consider the case $d=1$ for simplicity.

Our result of Corollary \ref{cor:submodular} has some connection with a theorem  by Hao-Jog as we now explain.

Motivated by the monotonicity of the Shannon entropy 
with respect to the number of summands in the central limit theorem, and related conjectures \cite{BNT16,ENT18:1}, in \cite{HJ18:isit}
the authors proved the following inequality for symmetric\footnote{Symmetric means that the density $p \colon (\mathbb{R}^d)^n \to \mathbb{R}^+$ of $X$ satisfies $p(\varepsilon x)=p(x)$ for all $x \in (\mathbb{R}^d)^n$ and all $ \varepsilon \in (\{-1,1\}^d)^n$ where $\varepsilon x$ is the product coordinate by coordinate.} random vector $X=(X_1,\dots,X_n)$
$$
h \left( \frac{\sum_{i=1}^n X_i}{\sqrt{n}} \right)
\geq \frac{h(X)}{n} .
$$
Assume that the density $p$ of $X$ is log-supermodular. By Corollary \ref{cor:submodular} and the fact that, for any $d$ dimensional random variable $Y$ and any $a>0$, $h(aY)=h(Y)+d\log a$,
it holds
\begin{align*}
& h \left( \frac{\sum_{i=1}^n X_i}{\sqrt{n}} \right)
 =
h \left(\sum_{i=1}^n X_i \right) - \frac{d}{2}\log n \\
& \geq 
\frac{d}{2} \log \left(\frac{e^{\frac{2}{d}h(X_1|X_i, i \neq 1)} + \dots +  e^{\frac{2}{d}h(X_n|X_i, i \neq n)}}{n} \right) \\
& \geq
\frac{1}{n} \left( h(X_1|X_i, i \neq 1) + \dots +  h(X_n|X_i, i \neq n) \right)  
\end{align*}
where the last line follows from the concavity of the logarithm.
The right hand side of the latter is known to be the erasure entropy $\bar h(X)$ of $X$ \cite{VW06:isit} that is less that $h(X)$. 

As a consequence our result of Corollary \ref{cor:submodular} does not directly compare to that of Hao and Jog but both imply
$$
h \left( \frac{\sum_{i=1}^n X_i}{\sqrt{n}} \right)
\geq \frac{\bar h(X)}{n} 
$$
under different assumptions (symmetry or log-supermodularity).

\noindent
{\bf{Acknowledgements:}}  The authors thank Oliver Johnson for fruitful discussion and helpful comments, as well  as two anonymous reviewers whose careful reading and suggestions improved the quality of this paper. The second author's work was supported in part by SECIHTI grant CBF-2024-2024-3907.

\bibliographystyle{plain}
\bibliography{pustak} %conditional-EPI

\end{document}